\documentclass[11pt, reqno]{amsart}
\usepackage{eurosym}
\usepackage{graphicx,amsmath,amssymb,epsfig, float}
\usepackage{amsfonts, lscape}
\usepackage{color}
\usepackage{bbm,threeparttable,subcaption}
\usepackage{array}
\usepackage{verbatim}

\theoremstyle{plain}

\newtheorem*{thmUN}{Theorem}
\newtheorem*{propUN}{Proposition}

\theoremstyle{definition}

\newtheorem{example}{Example}

\theoremstyle{definition}

\graphicspath{ {FigsEtc/} }
\newcolumntype{M}[1]{>{\centering\arraybackslash}m{#1}}
\input{tcilatex}

\begin{document}
\def\sym#1{\ifmmode^{#1}\else\(^{#1}\)\fi}

\title{Fritz John's equation in mechanism design}
\author{Alfred Galichon{$^{\S }$}}
\date{October 5, 2020. Funding from NSF grant DMS-1716489, as well as ERC grant
CoG-866274 are acknowledged. The author benefited from insightful
discussions with Guillaume Carlier, Deniz Dizdar, and Benny Moldovanu.\\
{\indent$^{\S }$New York University, departments of economics and
mathematics and Sciences Po, department of economics; ag133@nyu.edu}}

\begin{abstract}
We show the role that an important equation first studied by Fritz John
plays in mechanism design.

\textbf{Dedicated to Nicholas Yannelis on his 65th birthday.}

\vspace{0.1cm} \emph{Keywords:} implementability, mechanism design, John's
equation, Kevin Roberts' theorem

\vspace{0.1cm}
\end{abstract}

\maketitle

\pagebreak

A large part of the literature on mechanism design deals with
implementability in dominant strategy. Let us recall the basic result in the
single-agent case, following Rochet (1987) and McAfee and McMillan (1988),
and as exposited in Chapter 4.4 of Vohra (2011). Assume $x\in \mathbb{R}^{d}$
is the type reported by the agent, and $z\in \mathbb{R}^{d}$ is the outcome
selected by the mechanism. The mechanism specified an allocation rule $T:%
\mathbb{R}^{d}\rightarrow \mathbb{R}^{d}$ and a payment rule $\pi :\mathbb{R}%
^{d}\rightarrow \mathbb{R}$. If the agent announces $x$, the outcome $%
z=T\left( x\right) $ is selected, while the agent is asked to pay $\pi
\left( x\right) $. It is assumed that if the agent is of type $x$, if
outcome $z$ is selected, and if the payment is $\pi $, the agent's utility
is $x^{\top }z-\pi $. The mechanism is called implementable in dominant
strategy (or simply implementable) if reporting her true type is the agent's
dominant strategy; an allocation rule $T$ is called implementable in
dominant strategy if there exists a payment rule $\pi $ such that the
mechanism $(T,\pi )$ is implementable. This happens if
\begin{equation*}
x^{\top }T\left( x\right) -\pi \left( x\right) \geq x^{\top }T\left(
x^{\prime }\right) -\pi \left( x^{\prime }\right) ~\forall x^{\prime }\in
\mathbb{R}^{d}.
\end{equation*}

Denoting $V\left( x\right) =\max_{x^{\prime }\in \mathbb{R}^{d}}\left\{
x^{\top }T\left( x^{\prime }\right) -\pi \left( x^{\prime }\right) \right\} $%
, this will be the case when $T\left( x\right) $ is in the subdifferential
of $V\left( x\right) $, or when $T$ is continuous, when $T\left( x\right)
=\nabla u\left( x\right) $.

Hence the following result due to Rochet (1987) and McAfee and McMillan
(1988):

\begin{thmUN}[Implementation theorem]
In the single-agent case, a continuous allocation rule $T:\mathbb{R}%
^{d}\rightarrow \mathbb{R}^{d}$ is implementable in dominant strategy if and
only if $T\left( x\right) =\nabla V\left( x\right) $ for some convex
function $V$.
\end{thmUN}

\bigskip

The purpose of this note is to investigate the multi-agent case. Assume that
the space of types of each agent is still $\mathbb{R}^{d}$, and denote $x\in
\mathbb{R}^{d}$ the type of the first agent and $y\in \mathbb{R}^{d}$ the
type of the second agent. The outcome $z$ is still an element of $\mathbb{R}%
^{d}$, and the allocation rule is now a map $T:\mathbb{R}^{d}\times \mathbb{R%
}^{d}\rightarrow \mathbb{R}^{d}$, where $z=T\left( x,y\right) $ is the
outcome selected if agent 1 announces type $x$ and agent 2 announces type $y$%
. The payment by agent 1 is $\pi _{1}\left( x,y\right) $ while the payment
by agent 2 is $\pi _{2}\left( x,y\right) $. Denoting $V_{1}\left( x,y\right)
=\max_{x^{\prime }\in \mathbb{R}^{d}}\left\{ x^{\top }T\left( x^{\prime
},y\right) -\pi _{1}\left( x^{\prime },y\right) \right\} $ and $V_{2}\left(
x,y\right) =\max_{y^{\prime }\in \mathbb{R}^{d}}\left\{ y^{\top }T\left(
x,y^{\prime }\right) -\pi _{2}\left( x,y^{\prime }\right) \right\} $, it is
easy to adapt the previous theorem to show that in the two-agent case, a
continuous allocation rule $T:\mathbb{R}^{d}\times \mathbb{R}^{d}\rightarrow
\mathbb{R}^{d}$ is implementable in dominant strategy if and only if $%
T\left( x,y\right) =\nabla _{x}V_{1}\left( x,y\right) $ for some function $%
V_{1}\left( x,y\right) $ which is convex in $x$ for all $y$, and $T\left(
x,y\right) =\nabla _{y}V_{2}\left( x,y\right) $ for some function $%
V_{2}\left( x,y\right) $ which is convex in $y$ for all $x$.

\bigskip

The main result in this note is the following statement:

\begin{propUN}
Consider a smooth allocation rule $T:\mathbb{R}^{d}\times \mathbb{R}%
^{d}\rightarrow \mathbb{R}^{d}$, and assume it is implementable. Then $%
T\left( x,y\right) =\nabla _{x}V_{1}\left( x,y\right) $ where $V_{1}$
satisfies \emph{Fritz John's equation}%
\begin{equation}
\frac{\partial ^{2}V_{1}\left( x,y\right) }{\partial x_{i}\partial y_{j}}=%
\frac{\partial ^{2}V_{1}\left( x,y\right) }{\partial x_{j}\partial y_{i}}%
,1\leq i,j\leq d  \label{JohnEqn}
\end{equation}%
and in addition, the resulting symmetric matrix is semidefinite positive.
Similarly, $T\left( x,y\right) =\nabla _{y}V_{2}\left( x,y\right) $ where $%
V_{2}$ satisfies the same restrictions.
\end{propUN}

\begin{proof}
If $T$ is implementable, then $T\left( x,y\right) =\nabla _{x}V_{1}\left(
x,y\right) $, where $V_{1}\left( x,y\right) $ is convex in $x$ for all $y$
and $T\left( x,y\right) =\nabla _{y}V_{2}\left( x,y\right) $ where $%
V_{2}\left( x,y\right) $ is convex in $y$ for all $x$. Because $T^{i}\left(
x,y\right) =\partial V_{2}\left( x,y\right) /\partial y_{i}$, one has $%
\partial T^{i}\left( x,y\right) /\partial y_{j}=\partial ^{2}V_{2}\left(
x,y\right) /\partial y_{i}\partial y_{j}$, and hence $\left( \partial
T^{i}\left( x,y\right) /\partial y_{j}\right) _{ij}$ is symmetric
semi-definite positive. But because $T$ is also a gradient with respect ot $%
x $, one has $T^{i}\left( x,y\right) =\partial V_{1}\left( x,y\right)
/\partial x_{i}$, and thus
\begin{equation*}
\frac{\partial ^{2}V_{1}}{\partial x_{i}\partial y_{j}}\left( x,y\right) =%
\frac{\partial T^{i}}{\partial y_{j}}\left( x,y\right) =\frac{\partial
^{2}V_{2}}{\partial y_{i}\partial y_{j}}\left( x,y\right) ,
\end{equation*}
which shows that $\left( \partial ^{2}V_{1}\left( x,y\right) /\partial
x_{i}\partial y_{j}\right) _{ij}$ is symmetric semi-definite positive.
Similarly, it is easy to see that%
\begin{equation*}
\frac{\partial ^{2}V_{2}}{\partial x_{i}\partial y_{j}}\left( x,y\right) =%
\frac{\partial T^{j}}{\partial x_{i}}\left( x,y\right) =\frac{\partial
^{2}V_{1}}{\partial x_{i}\partial x_{j}}\left( x,y\right) ,
\end{equation*}%
and therefore $\left( \partial ^{2}V_{2}\left( x,y\right) /\partial
x_{i}\partial y_{j}\right) _{ij}$ is also symmetric semi-definite positive.
\end{proof}

\bigskip

Equation~(\ref{JohnEqn}) is a well-known mathematical equation appearing in
harmonic analysis and inverse problems: it is called \emph{Fritz John's
ultrahyperbolic equation}, see John (1938), Kurusa (1991) and Ehrenpreis
(2003). It plays an important role in medical imagery because of its
connection with the so-called X-ray transform, a variant of the Radon
transform; however, to the best of the author's knowledge, its occurrence in
mechanism design problems seems to have remained unnoticed until now. Fritz
John (1938) for $d=3$, and Kurusa (1991) more generally provided rigorous
conditions under which the solutions to~(\ref{JohnEqn}) are given exactly by
functions of the form%
\begin{equation}
V_{1}\left( x,y\right) =\int_{-\infty }^{+\infty }\frac{1}{\lambda }\phi
_{\lambda }\left( \lambda x+\left( 1-\lambda \right) y\right) d\lambda
\label{xrayForm}
\end{equation}%
where $\phi _{\lambda }:\mathbb{R}^{d}\rightarrow \mathbb{R}$. Indeed,
\begin{equation*}
\frac{\partial ^{2}\phi _{\lambda }\left( \left( 1-\lambda \right) x+\lambda
y\right) }{\partial x_{i}\partial y_{j}}=\lambda \left( 1-\lambda \right)
\frac{\partial ^{2}\phi _{\lambda }}{\partial w_{i}\partial w_{j}}\left(
\left( 1-\lambda \right) x+\lambda y\right)
\end{equation*}%
is symmetric, and thus the sum is.

\bigskip

Note, however that while functions of the form~(\ref{xrayForm}) satisfy
John's equation~(\ref{JohnEqn}), they do not necessarily satisfy the
positive semidefiniteness restriction that are expressed in the
proposition.\ In order to ensure this restriction is satisfied, it is
natural to restrict to $\lambda \in \left[ 0,1\right] $ and $\phi _{\lambda
} $ convex, and thus introduce the class of solutions%
\begin{equation*}
V_{1}\left( x,y\right) =\int_{0}^{1}\frac{1}{\lambda }\phi _{\lambda }\left(
\lambda x+\left( 1-\lambda \right) y\right) d\lambda
\end{equation*}%
where $\phi _{\lambda }:\mathbb{R}^{d}\rightarrow \mathbb{R}$ are convex
functions. This yields solutions of the form%
\begin{eqnarray*}
T\left( x,y\right) &=&\int_{0}^{1}T_{\lambda }\left( x,y\right) d\lambda , \\
&&\text{where} \\
T_{\lambda }\left( x,y\right) &:&=\nabla \phi _{\lambda }\left( \lambda
x+\left( 1-\lambda \right) y\right) ,
\end{eqnarray*}%
and $T_{\lambda }\left( x,y\right) $ is called an \emph{elementary
allocation rule}.

\bigskip

Let us study the elementary allocation rules $T_{\lambda }\left( x,y\right) $%
. One has
\begin{equation*}
\nabla \phi _{\lambda }\left( w\right) =\arg \max_{z\in \mathbb{R}%
^{d}}\left\{ w^{\top }z-\phi _{\lambda }^{\ast }\left( z\right) \right\} ,
\end{equation*}%
where $\phi _{\lambda }^{\ast }\left( z\right) =\max_{w\in \mathbb{R}%
^{d}}\left\{ w^{\top }z-\phi _{\lambda }\left( w\right) \right\} $ can be
interpreted as a payment rule. Hence,%
\begin{equation*}
T_{\lambda }\left( x,y\right) =\arg \max_{z\in \mathbb{R}^{d}}\left\{
\lambda x^{\top }z+\left( 1-\lambda \right) y^{\top }z-\phi _{\lambda
}^{\ast }\left( z\right) \right\} .
\end{equation*}

Note that $\lambda x^{\top }z+\left( 1-\lambda \right) y^{\top }z$ is a
measure of the social welfare where one assigns weight $\lambda $ to agent
1, and weight $\left( 1-\lambda \right) $ to agent 2. Therefore, $T_{\lambda
}\left( x,y\right) $ is an affine welfare maximizer. Note that when one
imposes further that the set of outcomes should be finite and when $d\geq 2$%
, a theorem by Kevin Roberts (1979) asserts that the only possible
allocation rule should be the affine welfare maximizers \footnote{%
Jehiel et al. (2008) study the notion of \emph{cardinal potential} in the
context of ex-post implementability, and derive a related partial
differential equation which also bears a connection with Roberts' theorem,
although they don't make the link with Fritz John's equation.}. Removing the
restriction that the set of outcomes should be finite yields many more
solutions -- in particular, sums of affine welfare maximizers. A problem
that seems interesting is to determine if when $d\geq 2$, there are
implementable rules that are not affine welfare maximizers.

\bigskip

Let us take a very simple example:

\begin{example}
Consider a situation where two goods must be allocated between two players,
so that each player gets one good. Player $1$ has valuation $x_{1}$ for good
$1$ and $x_{2}$ for good $2$, and player $2$ has valuation $y_{1}$ for good $%
1$, and $y_{2}$ for good $2$. It is assumed that $x_{1}>x_{2}$ and $%
y_{1}<y_{2}$. Call \textquotedblleft direct\textquotedblright\ the
assignment where player $1$ gets good $1$ and player $2$ gets good $2$, and
\textquotedblleft reverse\textquotedblright\ the opposite assignment. Let $%
z_{1}$ be the probability of a direct assignment, and $z_{2}=1-z_{1}$ the
probability of a reverse assignment. The principal must decide on $z=\left(
z_{1},z_{1}\right) $ on the simplex. An implementable assignment rule is $%
z=T\left( x,y\right) $, where%
\begin{equation*}
T\left( x,y\right) =\left( \frac{x_{1}-x_{2}}{x_{1}-x_{2}+y_{2}-y_{1}},\frac{%
y_{2}-y_{1}}{x_{1}-x_{2}+y_{2}-y_{1}}\right) .
\end{equation*}%
indeed, letting
\begin{equation*}
V_{1}\left( x,y\right) =x_{1}-\left( y_{2}-y_{1}\right) \log \left(
x_{1}-x_{2}+y_{2}-y_{1}\right) ,
\end{equation*}
one verifies that $V_{1}\left( x,y\right) $ is convex in $x$, and that $%
T=\nabla _{x}V_{1}$, while letting%
\begin{equation*}
V_{2}\left( x,y\right) =y_{2}-\left( x_{1}-x_{2}\right) \log \left(
x_{1}-x_{2}+y_{2}-y_{1}\right)
\end{equation*}%
one verifies that $V_{2}\left( x,y\right) $ is convex in $y$ and that $%
T=\nabla _{y}V_{2}$.

One has $\phi _{\lambda }\left( w\right) =\max \left\{ w_{1},w_{2}\right\} $
independent of $\lambda $, so that
\begin{equation*}
\nabla \phi _{\lambda }\left( w\right) =\left( 1\left\{ w_{1}\geq
w_{2}\right\} ,1\left\{ w_{1}<w_{2}\right\} \right) ,
\end{equation*}%
and when $w=\lambda x+\left( 1-\lambda \right) y$, one has%
\begin{equation*}
T_{\lambda }\left( x,y\right) =\left(
\begin{array}{c}
1\left\{ \lambda \left( x_{1}-x_{2}\right) +\left( 1-\lambda \right) \left(
y_{1}-y_{2}\right) \geq 0\right\} , \\
1\left\{ \lambda \left( x_{1}-x_{2}\right) +\left( 1-\lambda \right) \left(
y_{1}-y_{2}\right) <0\right\}%
\end{array}%
\right)
\end{equation*}%
and thus, integrating over $\lambda \in \left[ 0,1\right] $ with respect to
the Lebesgue measure,%
\begin{equation*}
T\left( x,y\right) =\int_{0}^{1}T_{\lambda }\left( x,y\right) d\lambda
\end{equation*}%
that is%
\begin{equation*}
T\left( x,y\right) =\left( \frac{x_{1}-x_{2}}{x_{1}-x_{2}+y_{2}-y_{1}},\frac{%
y_{2}-y_{1}}{x_{1}-x_{2}+y_{2}-y_{1}}\right) .
\end{equation*}

This assignment rule can be interpreted as follows:

Draw $\lambda $ uniformly from $\left[ 0,1\right] $. Scale the valuation of
player $1$ by $\lambda $, and the valuations of player $2$ by $\left(
1-\lambda \right) $. Compute the valuation after rescaling associated with
the direct and reverse assignment, respectively. Play the assignment which
has whichever higher valuation.
\end{example}


\begin{thebibliography}{9}
\bibitem{Ehrenpreis} Ehrenpreis, L. (2003). \emph{The Universality of the
Radon Transform}. Oxford University Press.

\bibitem{JMM} Jehiel, P., Moldovanu, B., and Meyer-ter-Veh, M. (2008).
\textquotedblleft Ex-post Implementation and Preference Aggregation via
Potentials\textquotedblright . \emph{Economic Theory} 37 (3), 469 -- 490.

\bibitem{John} John, F. (1938). \textquotedblleft The ultrahyperbolic
differential equation with four independent variables\textquotedblright .
\emph{Duke Mathematical Journal} 4 (2), pp. 300--322.

\bibitem{Kurusa} Kurusa, \'{A}. (1991). \textquotedblleft A characterization
of the Radon transform's range by a system of PDEs.\textquotedblright\ \emph{%
Journal of Mathematical Analysis and Applications}, 161 (1), pp. 218--226.

\bibitem{McAfeeMcMillan} McAfee P. and McMillan, J. (1988).
\textquotedblleft Multidimensional incentive compatibility and mechanism
design.\textquotedblright\ \emph{Journal of Economic Theory} 46 (2), pp.
335--354.

\bibitem{Roberts} Roberts, K. (1979). \textquotedblleft The characterization
of implementable choice rules\textquotedblright . In Jean-Jacques Laffont,
editor, \emph{Aggregation and Revelation of Preferences}. Papers presented
at the first European Summer Workshop of the Economic Society, pages
321--349. North-Holland.

\bibitem{Rochet} Rochet, J.-C. (1987). \textquotedblleft A necessary and
sufficient condition for rationalizability in a quasi-linear
context.\textquotedblright\ \emph{Journal of Mathematical Economics} 16 (2),
pp. 191--200.

\bibitem{Vohra2011} Vohra, R. (2011). \emph{Mechanism Design: A Linear
Programming Approach}. Cambridge University Press.
\end{thebibliography}
\end{document}